\documentclass[12pt, reqno]{amsart}
\usepackage{amsmath, amsthm, amscd, amsfonts, amssymb, graphicx, color}
\usepackage{graphicx}
\usepackage{epstopdf}
\usepackage{epsfig}
\usepackage{subfigure}
\usepackage[bookmarksnumbered, colorlinks, plainpages]{hyperref}
\allowdisplaybreaks[4]
\hypersetup{colorlinks=true,linkcolor=red, anchorcolor=green, citecolor=cyan, urlcolor=red, filecolor=magenta, pdftoolbar=true}

\newtheorem{theorem}{Theorem}[section]

\newtheorem{corollary}[theorem]{Corollary}
\theoremstyle{definition}
\newtheorem{definition}[theorem]{Definition}

\theoremstyle{remark}

\numberwithin{equation}{section}

\begin{document}

\setcounter{page}{1}

\title[Mutually-commuting von Neumann algebra models of quantum networks]{Mutually-commuting von Neumann algebra models of quantum networks and violation of Bell-type inequalities}

\author{Shuyuan Yang, Jinchuan Hou and Kan He}

\address[Shuyuan Yang]{Department of Mathematics, North University of China, TaiYuan, People's Republic of China}\address{School of Cyberspace Science and Technology, Beijing Institute of Technology, Beijing,
China}
 \email{yangshuyuan2000@163.com}

\address[Jinchuan Hou]{Department of Mathematics, Taiyuan University of Technology, TaiYuan, People's Republic of China.}

 \email{jinchuanhou@aliyun.com}

\address[Kan He]{Department of Mathematics, Taiyuan University of Technology, TaiYuan, People's Republic of China.}

 \email{hekan@tyut.edu.cn}


\begin{abstract}

Employing mutually-commuting von Neumann algebras to represent the algebra of observables on quantum systems provides a framework for studying quantum information theory in systems with infinite degrees of freedom and quantum field theory, yielding many profound results that differ from non-relativistic quantum systems. In this paper, we establish a mutually-commuting von Neumann algebra model of quantum networks with arbitrary structures. We derive Bell-type inequalities on this model, and determine various bounds for Bell-type inequalities based on the structure of underline von Neumann algebras, and identify the algebraic
structural conditions required for their violation. The conditions on the algebraic structure of observables for maximal violation of Bell-type inequalities, which we discovered in the context of von Neumann algebra models, can in turn guide the search for measurements in the non-relativistic setting.

\end{abstract} \maketitle

\section{\textbf{Introduction}}

Motivated by quantum field theory, which originates from the study of relativistic quantum mechanics, many novel quantum phenomena in systems with infinitely many degrees of freedom have been discovered \cite{HRDK,FK,ASR,MIPRE,CB1,STal,VLAS}. This differs from the non-relativistic quantum-mechanical setup, which is usually linked to type I von Neumann algebras and relies on the algebraic tensor product as its mathematical framework \cite{KR,LNP,SJ}. These are two distinct models, referred to respectively as the tensor product algebra (TPA) model and the mutually-commuting  von Neumann (observable) algebra (MCvNA) model. In the MCvNA model, there is, in general, no tensor product decomposition of the Hilbert space describing subsystems.

In the  TPA model, a bipartite composite quantum system is described by the tensor product of two Hilbert spaces $\mathcal{H}_\mathcal A$ and $\mathcal{H}_\mathcal B$, denoted as 
$\mathcal{H} = \mathcal{H}_\mathcal A \otimes \mathcal{H}_\mathcal B.$
Let $\mathcal{B}(\mathcal{H})$ denote the algebra of all bounded linear operators on $\mathcal{H}$, and $\mathcal{T}(\mathcal{H})$ denote the space of trace-class operators. Then there is a duality relation  $\mathcal{B}(\mathcal{H}) = \mathcal{T}(\mathcal{H})^*$. Therefore,  
a quantum state $\rho$ is a positive semidefinite operator with trace 1 in $ \mathcal{T}(\mathcal{H})$.  
The algebra of observables is the Jordan algebra of self-adjoint operators on $\mathcal B(\mathcal{H}).$ 
However, it should be pointed out that relying solely on the TPA model to discuss quantum information problems has two drawbacks \cite{RLJ,FCK,WE}:
1) it fails to provide a universal framework for accurately describing phenomena in systems with infinite degrees of freedom;
2) it is not suitable for studying traditional quantum physics such as quantum field theory, which requires the language of type III von Neumann algebras. Research on quantum information problems on von Neumann algebras has received significant attention and yielded many meaningful results from a mathematical perspective \cite{Yinzhi1,Gaoli1,Gaoli2,BDE1,KJD1,MM1,CKL1,JL1, wujinsong1, wujinsong2}. 
In MCvNA model, the algebra of observables of  quantum systems is described by a von Neumann algebra $\mathcal{M}$, with $\mathcal{M}_\mathcal A$ and $\mathcal{M}_\mathcal B$ being two mutually commuting von Neumann subalgebras of $\mathcal{M}$ such that $(\mathcal{M}_\mathcal A \vee \mathcal{M}_\mathcal B)''=\mathcal{M}$.
Here, $\mathcal{M}''$ denotes the double commutant of $\mathcal{M}$. However, in more general cases, a natural question arises: are the two models isomorphic? In other words, does $\mathcal{M}$ admit an isomorphism to $\mathcal{M}_\mathcal A \otimes \mathcal{M}_\mathcal B$? This question is intimately related to the famous Tsirelson problem, which is equivalent to the Connes embedding conjecture in von Neumann algebra theory \cite{MIPRE,JMe,ON}.
The Tsirelson problem asks whether the sets of joint probability distributions of measurement outcomes are the same in the two models. In 2020, Ji et al. \cite{MIPRE} designed a quantum interactive protocol based on the classical halting problem in computability theory, and proved that the answer to the Tsirelson problem is negative. This means that the TPA model and the MCvNA model are not equivalent. Therefore, the MCvNA model provides a more general framework.

In  non-relativistic quantum mechanics, Bell nonlocality demonstrates that local measurements performed on one subsystem of a quantum state can instantaneously influence the measurement outcomes on another subsystem, regardless of the spatial separation between them \cite{SCA, JSB, JS1}. Such nonlocal correlations can be detected through Bell inequalities, which serve as constraints that all local correlations must obey \cite{ABN, CJL, Guo, CG, HSA}. It has been demonstrated to offer quantum advantages in various device-independent quantum information tasks, including communication complexity \cite{RH}, quantum key distribution \cite{JLA, LSA}, randomness amplification \cite{S, RAA}, and measurement-based quantum computation \cite{RH1, RDH}. In the early 1980s, Summers et al. first introduced the maximal violation of Bell inequalities and prove that its value is bounded by 2$\sqrt{2}$ in the MCvNA model, with equality attainable iff each subalgebra contains a copy of $\mathcal{M}_2({\Bbb C})$ \cite{summer1}. This shows that Bell nonlocality is not merely a quantum peculiarity but a structural feature encoded in the classification of operator algebras, providing rigorous tools to quantify non-classical correlations in relativistic quantum systems \cite{SR}. Translating these bounds into the vacuum representation of algebraic quantum field theory, they show that tangent wedge algebras are always maximally correlated, whereas strictly spacelike-separated wedges decay exponentially with mass-governed distance \cite{summer2,summer3,summer4,summer5,summer6}.  
These works reveal a novel algebraic invariant, termed the Bell correlation invariant, which distinguishes infinitely many isomorphism classes of pairs of mutually commuting von Neumann algebras and links the maximal violation to the occurrence of the hyperfinite type $\rm II_1$ factor \cite{SR}. This is a pioneering work to make Bell nonlocality in quantum field theory (QFT) serve as a crucial bridge connecting quantum information science with fundamental physics \cite{SR,SBA}. It provides a rigorous framework for reconciling quantum entanglement with relativistic causality, resolves conceptual challenges such as impossible measurements, and reveals how fundamental symmetries like parity violation affect quantum correlations \cite{DHL,FCR,NY,KSA,MY}.

In contrast to entanglement originating from an individual source, quantum networks comprise numerous small-scale entangled states. Owing to the independence among distinct sources, the correlations emerging from quantum networks exhibit non-convex characteristics that transcend the polytopes associated with single-source entanglement \cite{nc1, nc2, nc3, nc4, nc5,Tava}. To date, Bell-like inequalities in the non-relativistic quantum mechanics have been devised to certify nonlocal correlations across diverse network architectures, such as entanglement-swapping networks \cite{nc1, 5}, chain configurations \cite{8, AMID}, star networks \cite{9, 10}, polygon structures \cite{11, 13,npj2}, tree-shaped networks \cite{14, 15, 16}, arbitrary acyclic networks \cite{nc3, nc4, any3}, and arbitrary $k$-independent networks \cite{12}. Alternative research directions examine the stronger forms of network nonlocality that surpass hybrid implementations involving classical variables and post-quantum resources \cite{fnn, LYK}. Nevertheless, limited progress has been made concerning the discrimination of correlations produced by different networks and the subsequent identification of underlying quantum network topologies \cite{npj1}.
Recently, the notion of bilocality in an  entanglement swapping network based on the MCvNA model has already been introduced by Ligthart L. T., et al \cite{LLD} and  Renou et. al \cite{LLD1}, and Xu has addressed the inclusion problem between TPA model  and MCvNA model in this setting \cite{XuXiang}. However, the MCvNA model of arbitrary networks and Bell-type inequalities on them have not yet been established now. In this paper, we aim to establish the MCvNA model of quantum networks with arbitrary structures and Bell-type inequalities on this model, and show its distinction from the non-relativistic case.

\section{Mutually-commuting von Neumann algebra models of quantum networks}

In quantum field theory, the observables performed by each party always correspond to a von Neumann algebra. We assume that party ${\bf A}_i$ corresponds to the algebra $\mathcal M_{A_i}$ for any $i\in\{1,2,...,m\}$. 
Building on this foundation, we have established the following mutually commuting von Neumann algebra model.

\begin{definition} {\bf (Mutually commuting von Neuamnn algebra model of multipartite quantum systems)} 
Let $\mathcal{M}_{A_i}$ be von Neumann subalgebra of $\mathcal B(H)$ over some Hilbert space $H$ for any $i\in\{1,2,...,m\}$, and they are mutually commuting, i.e., $\mathcal M_{A_i}\subset \mathcal M_{A_j}^\prime $ with $i\neq j\in \{1,2,...,m\}$, where $\mathcal M_{A_i}^\prime$ is the commutant of $\mathcal M_{A_i}$.  The generated von Neumann algebra $$\mathcal M_{A_1A_2\cdots A_m}=(\mathcal M_{A_1} \vee \mathcal M_{A_2}\vee \cdots \vee\mathcal M_{A_m})^{\prime\prime}.$$ 
We refer to the above model as the {\bf MCvNA} model of multipartite quantum systems. When $\mathcal M_{A_1A_2\cdots A_m}\simeq \mathcal M_{A_1} \otimes \mathcal M_{A_2} \otimes \cdots \otimes \mathcal M_{A_m}$, 
it is called the {\bf tensor product algebra model}. For any $A_{i}\in \mathcal M_{A_i}$ with $i\in\{1,2,...,m\}$, we always assume it is a self-adjoint element.
\end{definition}

Next, we define the mutually commuting von Neuamnn algebra model of quantum networks. Before that, let us review some basic concepts in non-relativistic quantum mechanics.  In a network $\Xi(n,m)$ with $m$ parties  ${\bf A}_1,{\bf A}_2,...,{\bf A}_m$ and $n$ sources, with the maximal independent number $h_{\rm max}$.  The independence number $h$ satisfies $2\leq h\leq h_{\rm max}$. In a quantum network, the independence number $h$ refers to the existence of $h$ parties that do not share the same sources, meaning that there must be no shared sources among these $h$ parties.  Let $\Lambda_i$ denote the set of hidden variables received by party ${\bf A}_i$. For each independence number $h$, there may be multiple choices for the set of corresponding independent parties. We denote it as $\Gamma(n,m,h^{D_h})$, which means the independence number $h$ corresponds to $D_h$ different sets of independent parties. We call \(D_{h}\) the degree of repetition.  For the purpose of illustration, a specific example is used below to explain the concept of independence of parties in a network.  For the quantum network shown in Fig. \ref{EX1}, which consists of 5 parties and 4 sources, it can be seen that parties ${\bf A}_1$, ${\bf A}_3$ and ${\bf A}_5$ do not share the same source. Therefore, the independence number of this network is 3. Obviously, there is no situation in this network where four parties do not share the same source, so the maximum independence number of this network is 3, i.e., $h_{\max}=3$. If $h=2$, then we have $D_2=6$ and $\Gamma(n,m,2^{D_2})=\{\{1,3\},\{1,4\},\{1,5\},\{2,4\},\{2,5\},\{3,5\}\}$.  We call $D_h$ is the degree of repetition of the indenpendence number $h$.

\begin{figure}[h]
  \centering
{\includegraphics[width=3.5in]{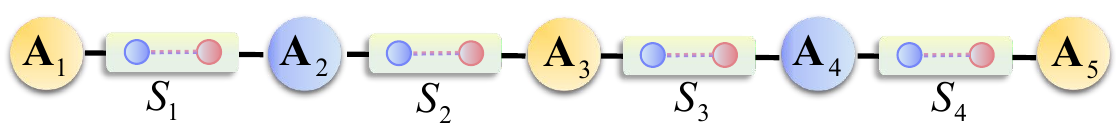}}
  \caption{  A network  with $h_{\rm max}=3$
   }\label{EX1}
\end{figure}

\begin{definition}\label{def2.2} {\bf (Mutually commuting von Neuamnn algebra models of quantum networks)}  For the network $\Xi(n,m)$ with $m$ parties  ${\bf A}_1,{\bf A}_2,...,{\bf A}_m$ and $n$ sources, with the maximal independent number $h_{\rm max}$, the mutually commuting von Neumann algebra model of $\Xi(n,m)$ is the multi-tuple $({\mathcal M_{A_i}}, \tau)$ satisfies the following conditions: \\
(1) $\{{\mathcal M_{A_i}}\}$ is a mutually commuting von Neuamnn algebra model of a $m$-partite quantum system;\\
(2) $\tau\in \mathcal M_{A_1A_2\cdots A_m}^*$ with $\mathcal M_{A_1A_2\cdots A_m}=(\mathcal M_{A_1} \vee \mathcal M_{A_2}\vee \cdots \vee\mathcal M_{A_m})^{\prime\prime}$, and satisfies $$\tau(A_{r_1}A_{r_2}\cdots A_{r_h})=\tau(A_{r_1})\tau(A_{r_2})\cdots \tau(A_{r_h}),$$
for all $(A_{r_1}, A_{r_2}, \cdots A_{r_h})\in \mathcal{M}_{A_{r_1}}\times \mathcal{M}_{A_{r_2}} \cdots \times\mathcal{M}_{A_{r_h}}$. Where  $2\leq h\leq h_{\rm max}$ and $\{r_1, r_2, \cdots, r_h\}\in\Gamma(n,m,h^{D_h})$.

We refer to the above model as the {\bf MCvNA} model of quantum networks.
\end{definition}
We call $\tau$ the network state corresponding to the network $\Xi(n,m)$, and this network state is not unique. For the network shown in Fig. \ref{EX1},  if $h=h_{\rm max}=3$, then the condition (2) in Definition \ref{def2.2}  can be concretized as $\tau\in \mathcal M_{A_1A_2\cdots A_5}^*$, and satisfies
\begin{align}
\tau(A_1A_3A_5)=\tau(A_1)\tau(A_3)\tau(A_5)
\end{align}
to be satisfied, where $A_1\in\mathcal{M}_{A_1}$, $A_3\in\mathcal{M}_{A_3}$ and 
$A_5\in\mathcal{M}_{A_5}$.

We point out that the von Neumann algebra model defined above is intended for the study of algebraic structural properties. Therefore, we have adopted the most general definition, which encompasses but does not correspond to models of quantum mechanics or quantum field theory. In the following, we establish a Bell-type inequality for the $h$-independent networks based on this model.

\section{Bell-type inequalities and their violations}

For the non-relativistic case, in the network $\Xi(n,m)$, if each party performs measurements with binary inputs and binary outputs, we can always verify that the network can generate network nonlocal correlations by demonstrating the violation of a Bell-type inequality. 
Specifically, suppose party ${\bf A}_i$ performs observables $A_{i,0}$ and $A_{i,1}$ respectively. Then network nonlocal correlations can always be detected by violating the following Bell-type inequality\cite{12} 
\begin{align}\label{S1}
\mathcal S=|I|^{\frac{1}{h}}+|J|^{\frac{1}{h}}\le 2,
\end{align}
where
$$I=\left\langle \prod_{i\in\{r_1,r_2,\dots,r_h\}}(A_{i,0}+A_{i,1})\prod_{j\in\{1,2,\dots,m\}\setminus\{r_1,r_2,\dots,r_h\}}A_{j,0}\right\rangle,$$
and $$J=\left\langle\prod\limits_{i\in\{r_1,r_2,\dots,r_h\}}(A_{i,0}-A_{i,1})\prod_{j\in\{1,2,\dots,m\}\setminus\{r_1,r_2,\dots,r_h\}}A_{j,1}\right\rangle.$$ 
Here, $h$ represents the independence number of the network, i.e., there exist $h$ parties in the network that do not share common sources. Meanwhile, it is assumed that the set of indicators for independent parties is $\{r_1,r_2,\dots,r_h\}$. Ref. \cite{12} also indicates that the maximum quantum violation is $2\sqrt{2}$.

In the MCvNA model  of the network $\Xi(n,m)$ with the maximal independent number $h_{\rm max}$, let $A_{i,{x_i}} \in \mathcal{M}_{A_i}$ be an observable  with spectrum ${1, -1}$. Thus, $A_{i,{x_i}}=\sum_{a_i=0}^1(-1)^{a_i}A_{i,a_i|x_i}$ with $a_i\in\{0,1\}$. We define
$$I_\tau=\tau\left(\prod_{i\in\{r_1,r_2,\dots,r_h\}}(A_{i,0}+A_{i,1})\prod_{j\in\{1,2,\dots,m\}\setminus\{r_1,r_2,\dots,r_h\}}A_{j,0}\right),$$
and $$J_\tau=\tau\left(\prod\limits_{i\in\{r_1,r_2,\dots,r_h\}}(A_{i,0}-A_{i,1})\prod_{j\in\{1,2,\dots,m\}\setminus\{r_1,r_2,\dots,r_h\}}A_{j,1}\right),$$ 
and 
\begin{eqnarray}\label{biin}
     \mathcal S_{\tau}=|I_\tau|^{\frac{1}{h}}+|J_\tau|^{\frac{1}{h}},
\end{eqnarray}
where $\{r_1,r_2,...,r_h\}\in\Gamma(n,m,h^{D_h})$ with $2\leq h\leq h_{\rm max}$ and
$$\tau(\prod_{i\in\{r_1,r_2,\dots,r_h\}}A_{i,{x_i}}\prod_{j\in\{1,2,\dots,m\}\setminus\{r_1,r_2,\dots,r_h\}}A_{j,{x_j}})=\sum_{a_1,a_2,...,a_m}(-1)^{\sum a_i}\tau(\Pi_{i=1}^mA_{i,a_i|x_i}).$$
The network correlation  $\hat{p}=p(a_1a_2\cdots a_m|x_1x_2\cdots x_m)$ can be defined as 
$$p(a_1a_2\cdots a_m|x_1x_2\cdots x_m)=\tau(\Pi_{i=1}^mA_{i,a_i|x_i}),$$
where  $x_i$ and $a_i$ denote the input and output of the measurement operator $A_{i,a_i|x_i}$, respectively, with $x_i,a_i\in{0,1}$ for each $i\in{1,\dots,m}$.

It should be noted that, in quantum mechanics, the maximal quantum violation of Bell-type inequality is $2\sqrt{2}$. In the MCvNA model of the network $\Xi(n,m)$, the supremum of $\mathcal S_{\tau}$ defined in Eq. (\ref{biin}) is $2\sqrt{2}$, coinciding with that in non-relativistic quantum mechanics. This result is proved in the following theorem.

\begin{theorem}\label{the3.1}
For the mutually commuting von Neumann algebra models of the quantum network $\Xi(n,m)$  with the maximal independent number $h_{\rm max}$, we always have
\begin{align}\label{eq3.2}
|I_\tau|^{\frac{1}{h}}+|J_\tau|^{\frac{1}{h}}\le2\sqrt{2},
\end{align}
where $2\leq h\leq h_{\rm max}$.
\end{theorem}
\begin{proof}
According to Gelfand-Namark-Segal (GNS) construction, there is a $^*$-representation $\pi_\tau: 
 \mathcal M_{A^1A^2\cdots A^m}\rightarrow \mathcal B(\mathcal H_\tau)$ and a cyclic vector $\Omega\in \mathcal B(\mathcal H_\tau)$ such that the set $\{\pi_\tau(R)\Omega:\  R\in \mathcal M_{A^1A^2\cdots A^m}\}$ is dense in $H_\tau$.
We have
\begin{align}
&\mathcal S_{\tau}=|I_\tau|^{\frac{1}{h}}+|J_\tau|^{\frac{1}{h}} \nonumber \\
=&\Biggl[ \Biggl|\tau\Biggl(\prod_{i\in\{r_1,r_2,\dots,r_h\}}(A_{i,0}+A_{i,1})\prod_{j\in\{1,2,\dots,m\}\setminus\{r_1,r_2,\dots,r_h\}}A_{j,0}\Biggr) \Biggr|
\Biggr]^{\frac{1}{h}} \nonumber \\
&+\Biggl[\Biggl|\tau\Biggl(\prod\limits_{i\in\{r_1,r_2,\dots,r_h\}}(A_{i,0}-A_{i,1})\prod_{j\in\{1,2,\dots,m\}\setminus\{r_1,r_2,\dots,r_h\}}A_{j,1}\Biggr)\Biggr|\Biggr]^{\frac{1}{h}} \nonumber\\
=& \Biggl|\Biggl\langle \pi_{\tau}\Biggl(\prod_{j\in\{1,2,\dots,m\}\setminus\{r_1,r_2,\dots,r_h\}}A_{j,0}\Biggr)\Omega, \pi_{\tau}\Biggl(\prod_{i\in\{r_1,r_2,\dots,r_h\}}(A_{i,0}+A_{i,1})\Biggr)\Omega
\Biggr\rangle\Biggr|^{\frac{1}{h}}
\nonumber \\
&+\Biggl|\Biggl\langle \pi_{\tau}\Biggl(\prod_{j\in\{1,2,\dots,m\}\setminus\{r_1,r_2,\dots,r_h\}}A_{j,1}\Biggr)\Omega, \pi_{\tau}\left(\prod_{i\in\{r_1,r_2,\dots,r_h\}}(A_{i,0}-A_{i,1})\right)\Omega
\Biggr\rangle \Biggr|
^{\frac{1}{h}} \nonumber \\
\le& \Biggl|\Biggl|\pi_{\tau}\Biggl(\prod_{j\in\{1,2,\dots,m\}\setminus\{r_1,r_2,\dots,r_h\}}A_{j,0}\Biggr)\Omega\Biggr|\Biggr|^{\frac{1}{h}}\,\Biggl|\Biggl| \pi_{\tau}\Biggl(\prod_{i\in\{r_1,r_2,\dots,r_h\}}(A_{i,0}+A_{i,1})\Biggr)\Omega \Biggr|\Biggr|^{\frac{1}{h}}
\nonumber \\
&+\Biggl|\Biggl| \pi_{\tau}\Biggl(\prod_{j\in\{1,2,\dots,m\}\setminus\{r_1,r_2,\dots,r_h\}}A_{j,1}\Biggr)\Omega\Biggr|\Biggr|^{\frac{1}{h}}\,\Biggl|\Biggl| \pi_{\tau}\Biggl(\prod_{i\in\{r_1,r_2,\dots,r_h\}}(A_{i,0}-A_{i,1})\Biggr)\Omega
\Biggr|\Biggr|^{\frac{1}{h}}  \label{E1} \\
\le&\Biggl|\Biggl| \pi_{\tau}\Biggl(\prod_{i\in\{r_1,r_2,\dots,r_h\}}(A_{i,0}+A_{i,1})\Biggr)\Omega \Biggr|\Biggr|^{\frac{1}{h}}+
\Biggl|\Biggl| \pi_{\tau}\Biggl(\prod_{i\in\{r_1,r_2,\dots,r_h\}}(A_{i,0}-A_{i,1})\Biggr)\Omega
\Biggr|\Biggr|^{\frac{1}{h}} \label{E2} \\
=& \tau\Biggl(\prod_{i\in\{r_1,r_2,\dots,r_h\}}(A_{i,0}+A_{i,1})^2\Biggr)^{\frac{1}{2h}}+\tau\Biggl(\prod_{i\in\{r_1,r_2,\dots,r_h\}}(A_{i,0}-A_{i,1})^2\Biggr)^{\frac{1}{2h}}.\label{eq3.3}
\end{align}

Take $F_{\tau}=\tau(\prod_{i\in\{r_1,r_2,\dots,r_h\}}(A_{i,0}+A_{i,1})^2)^{\frac{1}{2h}}+\tau(\prod_{i\in\{r_1,r_2,\dots,r_h\}}(A_{i,0}-A_{i,1})^2)^{\frac{1}{2h}}$. Thus, one can obtain
\begin{align}
F_{\tau}^2=&\tau\Biggl(\prod_{i\in\{r_1,r_2,\dots,r_h\}}(A_{i,0}+A_{i,1})^2\Biggr)^{\frac{1}{h}}+\tau\Biggl(\prod_{i\in\{r_1,r_2,\dots,r_h\}}(A_{i,0}-A_{i,1})^2\Biggr)^{\frac{1}{h}} \nonumber \\
&+2\tau\Biggl(\prod_{i\in\{r_1,r_2,\dots,r_h\}}(A_{i,0}+A_{i,1})^2\Biggr)^{\frac{1}{2h}}\tau\Biggl(\prod_{i\in\{r_1,r_2,\dots,r_h\}}(A_{i,0}-A_{i,1})^2\Biggr)^{\frac{1}{2h}}.
\end{align}
It  follows from the positivity property of $\tau$ that we have
\begin{align}
&\tau\Biggl(\prod_{i\in\{r_1,r_2,\dots,r_h\}}(A_{i,0}+A_{i,1})^2\Biggr)^{\frac{1}{h}} \nonumber \\=&\tau\Biggl(\prod_{i\in\{r_1,r_2,\dots,r_h\}}[(A_{i,0})^2+(A_{i,1})^2+A_{i,0}A_{i,1}+A_{i,1}A_{i,0}]\Biggr)^{\frac{1}{h}} \nonumber \\
\le&\tau\Biggl(\prod_{i\in\{r_1,r_2,\dots,r_h\}}(2I+A_{i,0}A_{i,1}+A_{i,1}A_{i,0})\Biggr)^{\frac{1}{h}}, \label{g35}
\end{align}
\begin{align}
&\tau\Biggl(\prod_{i\in\{r_1,r_2,\dots,r_h\}}(A_{i,0}-A_{i,1})^2\Biggr)^{\frac{1}{h}} \nonumber \\=&\tau\Biggl(\prod_{i\in\{r_1,r_2,\dots,r_h\}}[(A_{i,0})^2+(A_{i,1})^2-A_{i,0}A_{i,1}-A_{i,1}A_{i,0}]\Biggr)^{\frac{1}{h}} \nonumber \\
\le&\tau\Biggl(\prod_{i\in\{r_1,r_2,\dots,r_h\}}(2I-A_{i,0}A_{i,1}-A_{i,1}A_{i,0})\Biggr)^{\frac{1}{h}}, \label{g36}
\end{align}

By setting $2+\tau(A_{i,0}A_{i,1}+A_{i,1}A_{i,0})=4{\rm sin}^2\theta_i$ with $\theta_i\in[0,\pi]$, we obtain that $2-\tau(A_{i,0}A_{i,1}+A_{i,1}A_{i,0})=4{\rm cos}^2\theta_i$, $i\in\{r_1,r_2,\dots,r_h\}$.
From Ineqs. (\ref{g35}) and (\ref{g36}) we get that 
\begin{align}
&\tau\Biggl(\prod_{i\in\{r_1,r_2,\dots,r_h\}}(A_{i,0}+A_{i,1})^2\Biggr)^{\frac{1}{h}}+\tau\Biggl(\prod_{i\in\{r_1,r_2,\dots,r_h\}}(A_{i,0}-A_{i,1})^2\Biggr)^{\frac{1}{h}} \nonumber \\
=&\Biggl(\prod_{i\in\{r_1,r_2,\dots,r_h\}}4{\rm sin}^2\theta_i\Biggr)^{\frac{1}{h}}+\Biggl(\prod_{i\in\{r_1,r_2,\dots,r_h\}}4{\rm cos}^2\theta_i\Biggr)^{\frac{1}{h}} \nonumber \\
=&4\Biggl[ \Biggl(\prod_{i\in\{r_1,r_2,\dots,r_h\}}{\rm sin}\theta_i\Biggr)^{\frac{2}{h}}+\Biggl(\prod_{i\in\{r_1,r_2,\dots,r_h\}}{\rm cos}\theta_i\Biggr)^{\frac{2}{h}}
\Biggr] \nonumber \\
=&4\Biggl[ \Biggl(\prod_{i\in\{r_1,r_2,\dots,r_h\}}{\rm sin}\theta_i\Biggr)^{\frac{2}{h}}+\Biggl(\prod_{i\in\{r_1,r_2,\dots,r_h\}}{\rm sin}(\frac{\pi}{2}-\theta_i)\Biggr)^{\frac{2}{h}}
\Biggr] \nonumber \\
\le&4\Biggl[ {\rm sin}^2(\frac{1}{h}\sum_{i}\theta_i)+
 {\rm sin}^2(\frac{1}{h}\sum_{i}(\frac{\pi}{2}-\theta_i))
\Biggr] \label{sin}\\
=&4\Biggl[ {\rm sin}^2(\frac{1}{h}\sum_{i}\theta_i)+
 {\rm sin}^2(\frac{\pi}{2}-\frac{1}{h}\sum_{i}\theta_i)
\Biggr] \nonumber \\
=&4\Biggl[ {\rm sin}^2(\frac{1}{h}\sum_{i}\theta_i)+
 {\rm cos}^2(\frac{1}{h}\sum_{i}\theta_i)
\Biggr] \nonumber \\
=&4. \label{g38}
\end{align}
Here, the reason why the Ineq. $(\ref{sin})$ holds here is that for any $\theta_{r_1}, \theta_{r_2},...,\theta_{r_h}$ and integer $2\le h$, inequality $(\prod_{i\in\{r_1,r_2,\dots,r_h\}}{\rm sin}\theta_i)^{\frac{1}{h}}\le{\rm sin}(\frac{1}{h}\sum_{i\in\{r_1,r_2,\dots,r_h\}}\theta_i)$ holds \cite{12}.

It follows from the H\"older inequality 
that we get
\begin{align}
&  2\tau\Biggl(\prod_{i\in\{r_1,r_2,\dots,r_h\}}(A_{i,0}+A_{i,1})^2\Biggr)^{\frac{1}{2h}}\tau\Biggl(\prod_{i\in\{r_1,r_2,\dots,r_h\}}(A_{i,0}-A_{i,1})^2\Biggr)^{\frac{1}{2h}}\nonumber \\
\le&\tau\Biggl(\prod_{i\in\{r_1,r_2,\dots,r_h\}}(A_{i,0}+A_{i,1})^2\Biggr)^{\frac{1}{h}}+\tau\Biggl(\prod_{i\in\{r_1,r_2,\dots,r_h\}}(A_{i,0}-A_{i,1})^2\Biggr)^{\frac{1}{h}} \label{E3} \\
\le& \prod_{i\in\{r_1,r_2,\dots,r_h\}}(\tau[(A_{i,0}+A_{i,1})^2]+\tau[(A_{i,0}-A_{i,1})^2])^{\frac{1}{h}}
\nonumber\\
=& \prod_{i\in\{r_1,r_2,\dots,r_h\}}(2\tau[(A_{i,0})^2+(A_{i,1})^2])^{\frac{1}{h}}
\nonumber \\
\le&\prod_{i\in\{r_1,r_2,\dots,r_h\}}4^{\frac{1}{h}} \label{E7}\\
=&4 \label{g34e}
\end{align}
where the final inequality follows from the fact that the spectra of $A_{i,0}$ and $A_{i,1}$ are $-1$ and $1$, together with the positivity property of $\tau$.

Combining the Eqs. (\ref{g38}) and (\ref{g34e}), then we have 
\begin{align}
|I_\tau|^{\frac{1}{h}}+|J_\tau|^{\frac{1}{h}} \le \sqrt{F_{\tau}^2}\le\sqrt{4+4}=2\sqrt{2}.
\end{align}
\end{proof}

Below, for the a given quantum network $\Xi(n,m)$ with  the maximal independent number $h_{\rm max}$, we demonstrate how the value of $\mathcal S_{\tau}$ depends on the commutativity of the algebra families $\mathcal M_{A_1}, \mathcal M_{A_2},...,\mathcal M_{A_m}$ in the MCvNA model of the network $\Xi(n,m)$.

\begin{theorem}\label{the3.2}
For the mutually commuting von Neumann algebra models of the quantum network $\Xi(n,m)$  with the maximal independent number $h_{\rm max}$, if $\mathcal M_{A_{r_1}}$, $ \mathcal   M_{A_{r_2}},...,\mathcal M_{A_{r_h}}$ are Abelian, then 
\begin{align} \label{g315}
|I_\tau|^{\frac{1}{h}}+|J_\tau|^{\frac{1}{h}}\le2,
\end{align}
where $\{r_1,r_2,...,r_h\}\in\Gamma(n,m,h^{D_h})$ and $2\leq h \leq h_{\rm max}$.
\end{theorem}
\begin{proof}
For any $i\in\{r_1,r_2,...,r_h\}$, if $\mathcal M_{A_{i}}$ is Abelian,  then
\begin{align} \label{abelian}
a^{i}_{\epsilon_0\epsilon_1}\equiv\frac{1}{4}(1+\epsilon_0 A_{i,0})(1+\epsilon_1 A_{i,1})
\end{align}
with $\epsilon_0,\epsilon_1\in\{+,-\}$ is positive. Clearly, it follows from Eq. (\ref{abelian}) that we have
$$A_{i,0}+A_{i,1}=2(a^{i}_{++}-a^i_{--}),\,\,\,\,A_{i,0}-A_{i,1}=2(a^{i}_{+-}-a^{i}_{-+}).$$ 
Thus, we get
\begin{align}
&|I_\tau|^{\frac{1}{h}}+|J_\tau|^{\frac{1}{h}} \nonumber \\
=& \Biggl|\tau\Biggl(\prod_{i\in\{r_1,r_2,\dots,r_h\}}(A_{i,0}+A_{i,1})\prod_{j\in\{1,2,\dots,m\}\setminus\{r_1,r_2,\dots,r_h\}}A_{j,0}\Biggr) \Biggr|^{\frac{1}{h}} \nonumber \\
&+\Biggl|\tau\Biggl(\prod\limits_{i\in\{r_1,r_2,\dots,r_h\}}(A_{i,0}-A_{i,1})\prod_{j\in\{1,2,\dots,m\}\setminus\{r_1,r_2,\dots,r_h\}}A_{j,1}\Biggr)\Biggr|^{\frac{1}{h}} \nonumber\\
=&2\Biggl[ \Biggl|\tau\Biggl(\prod_{i\in\{r_1,r_2,\dots,r_h\}}(a^i_{++}+a^i_{--})\prod_{j\in\{1,2,\dots,m\}\setminus\{r_1,r_2,\dots,r_h\}}A_{j,0}\Biggr) \Biggr|^{\frac{1}{h}} \nonumber \\
&+\Biggl|\tau\Biggl(\prod\limits_{i\in\{r_1,r_2,\dots,r_h\}}(a^i_{+-}-a^i_{-+})\prod_{j\in\{1,2,\dots,m\}\setminus\{r_1,r_2,\dots,r_h\}}A_{j,1}\Biggr)\Biggr|^{\frac{1}{h}}\Biggl] \nonumber\\
\le&2\Biggl[  \Biggl( \prod\limits_{i,k\in\{r_1,r_2,\dots,r_h\},i\neq k}|\tau(a^i_{++}a^{k}_{++})+\tau(a^i_{++}a^{k}_{--})+\tau(a^i_{--}a^{k}_{++})+\tau(a^i_{--}a^{k}_{--})
|\Biggl)^{\frac{1}{h}}
\nonumber \\
&+\Biggl( \prod\limits_{i,k\in\{r_1,r_2,\dots,r_h\},i\neq j}|\tau(a^i_{+-}a^{k}_{+-})+\tau(a^i_{+-}a^{k}_{-+})+\tau(a^i_{-+}a^{k}_{+-})+\tau(a^i_{-+}a^{k}_{-+})
|\Biggl)^{\frac{1}{h}}\Biggl]\nonumber \\
=&2\Biggl[  \Biggl( \prod\limits_{i,k\in\{r_1,r_2,\dots,r_h\},k\neq j}|\tau(a^i_{++}+a^{i}_{--})||\tau(a^k_{--}+a^{k}_{++})
|\Biggl)^{\frac{1}{h}} \nonumber \\
&+\Biggl( \prod\limits_{i,k\in\{r_1,r_2,\dots,r_h\},k\neq j}|\tau(a^i_{+-}+a^{i}_{-+})||\tau(a^k_{+-}+a^{k}_{-+})
|\Biggl)^{\frac{1}{h}} \Biggl] \nonumber\\
\leq&2\prod\limits_{i,k\in\{r_1,r_2,\dots,r_h\},i\neq k}[(\tau(a^i_{++}+a^{i}_{--})+\tau(a^i_{+-}+a^{i}_{-+}))(\tau(a^k_{--}+a^{k}_{++})+\tau(a^k_{+-}+a^{k}_{-+}))]^{\frac{1}{h}}\nonumber \\
\leq&2\times\prod\limits_{i,k\in\{r_1,r_2,\dots,r_h\},i\neq k} 1\nonumber \\
=&2,
\end{align}
where the first inequality follows the fact that the spectra of $-I\le A_{i,0},A_{i,1}\le I$ and the order-preserving of state $\tau$ for any $j\in\{1,2,\dots,m\}\setminus\{r_1,r_2,\dots,r_h\}$, and the second inequality holds because of the H\"older inequality and the non-negativeness of $a^i_{\epsilon_0\epsilon_1}$.
\end{proof}

Based on the Theorem \ref{the3.2}, one can infer that the commutativity of the algebra families indeed affects the upper bound of the $\mathcal S_{\tau}$ shown in Eq. (\ref{biin}). It is worth noting that the independent parties in any given network are not unique. Once a set of independent parties is specified, it can always be shown through Theorem \ref{the3.2} that a violation of the Ineq. (\ref{g315}) implies that at least one of the algebra families $\{\mathcal M_{A_i}\}_{i\in\{r_1,r_2,\dots,r_h\}}$ is non-abelian. Here,  violation of the Ineq. (\ref{g315}) refers to $2<\mathcal S_{\tau}\leq 2\sqrt{2}$.

Furthermore, we attempt to analyze the influence of the choice of state on the value of $\mathcal S_{\tau}$. Before that, we first recall the notion of product state. Let $\mathcal{M}_{A_u}$ and $\mathcal{M}_{A_v}$ be von Neumann algebras. A state $\tau\in\mathcal{M}_{A_uA_v}^*$ is said to be product if there exist states $\xi_{\alpha}\in\mathcal{M}_{A_u}^*$ and $\eta_{\alpha}\in\mathcal{M}_{A_v}^*$
such that 
$$\tau(A_uA_v)=\xi_{\alpha}(A_u)\eta_{\alpha}(A_v)$$
for all $A_u\in\mathcal{M}_{A_u}$ and $A_v\in\mathcal{M}_{A_v}$.
If $\tau\in\mathcal{M}_{A_uA_v}^*$ is the closed convex hull of product states, then it is called separable, i.e., 
if there are $\xi_{\alpha_k}\in \mathcal M_{A_u^w}^*$ 
and $\eta_{\alpha_k}\in \mathcal M_{A_v^l}^*$
 such that  for  all $(A_u)^w_{i_w}\in \mathcal M_{A_u^w}$ and $(A_v)^l_{j_l}\in \mathcal M_{A_v^l}$,
$$\tau((A_u)^w_{i_w}(A_v)^l_{j_l})=\sum_{\lambda_{\alpha_k}}p(\lambda_{\alpha_k})\xi_{\alpha_k}((A_u)^w_{i_w})\eta_{\alpha_k}((A_v)^l_{j_l})$$ 
with $\sum_{\lambda_{\alpha_k}}p(\lambda_{\alpha_k})=1$,
then 
we call a state 
$\tau\in \mathcal \mathcal M_{A_uA_v}^*$  separable on $(\mathcal M_{A_u}\vee \mathcal M_{A_v})^{\prime\prime}$.  Otherwise, $\tau\in \mathcal \mathcal M_{A_uA_v}^*$ is called entangled.

In non-relativistic quantum mechanics, if the sources shared among the parties are all separable states  in a quantum network, then Ineq. (\ref{S1}) must hold. Corresponding to the MCvNA model of the quantum network, we assume that if sources are shared among the parties, then the restrictions of $\tau$ to these systems are all separable states, and the following theorem can be obtained.

\begin{theorem}\label{the3.3}
For the mutually commuting von Neuamnn algebra models of the quantum network $\Xi(n,m)$,
if  there exists a state $\tau$ on $\mathcal M_{A_1A_2\cdots A_m}$ such that,
whenever the $w$-th subsystem of party $A_u$ and the $l$-th subsystem of party $A_v$ share a source in the network $\Xi(n,m)$ for any $u,v\in\{1,2...,m\}$ and $u\neq v$,
the state $\tau|_{(\mathcal{M}_{\mathcal A_u^w}\vee\mathcal{M}_{\mathcal A_v^l})''}$ is separable,  then for any scheme with two inputs and two outputs, the observables $A_{i,x_{i}}$ $(i\in\{1,2,...,m\})$ satisfy 
$$\sup_{\{A_{i,0}, A_{i,1}\}_{i\in\{1,2,...,m\}}}\mathcal S_{\tau}=\sup_{\{A_{i,0}, A_{i,1}\}_{i\in\{1,2,...,m\}}}(|I_\tau|^{\frac{1}{h}}+|J_\tau|^{\frac{1}{h}})=2.$$

\end{theorem}

 \begin{proof}
 Using the separability assumption of $\tau$, we get 
 \begin{align}
&|I_\tau|^{\frac{1}{h}}+|J_\tau|^{\frac{1}{h}} \nonumber \\
=& \Biggl|\tau\Biggl(\prod_{i\in\{r_1,r_2,\dots,r_{h}\}}(A_{i,0}+A_{i,1})\prod_{j\in\{1,2,\dots,m\}\setminus\{r_1,r_2,\dots,r_{h}\}}A_{j,0}\Biggr) \Biggr|^{\frac{1}{h}} \nonumber \\
&+\Biggl|\tau\Biggl(\prod\limits_{i\in\{r_1,r_2,\dots,r_{h}\}}(A_{i,0}-A_{i,1})\prod_{j\in\{1,2,\dots,m\}\setminus\{r_1,r_2,\dots,r_{h}\}}A_{j,1}\Biggr)\Biggr|^{\frac{1}{h}} \nonumber\\
=&\Biggr| \prod_{k\in\{r_1,r_2,...,r_{h}\}}\sum_{\lambda_{\beta}\in\Lambda_k}p(\lambda_{\beta} )\xi_{\alpha_{\Lambda_k}}(A_{k,0}+A_{k,1})\eta_{\beta}(A_{k',0})\eta(X_0)
\Biggr|^{\frac{1}{h}}
\nonumber \\
&+
\Biggr| \prod_{k\in\{r_1,r_2,...,r_{h_1}\}}\sum_{\lambda_{\beta}\in\Lambda_k}p(\lambda_{\beta})\xi_{\alpha_{\Lambda_{k}}}(A_{k,0}-A_{k,1})\eta_{\beta}(A_{k',1})\eta(X_1)
\Biggr|^{\frac{1}{h}} \nonumber \\
\leq&\Biggr\{ \prod_{k\in\{r_1,r_2,...,r_{h_1}\}}\sum_{\lambda_{\beta}\in\Lambda_k}p(\lambda_\beta)|\xi_{\alpha_{\Lambda_k}}(A_{k,0}+A_{k,1})|
\Biggr\}^{\frac{1}{h}} \nonumber \\
&+
\Biggr\{ \prod_{k\in\{r_1,r_2,...,r_{h_1}\}}\sum_{\lambda_{\beta}\in\Lambda_k}p(\lambda_\beta)|\xi_{\alpha_{\Lambda_k}}(A_{k,0}-A_{k,1})|
\Biggr\}^{\frac{1}{h}} \nonumber \\
\leq&\prod_{k\in\{r_1,r_2,...,r_{h}\}}
\Biggr\{  \sum_{\lambda_{\beta}\in\Lambda_k}p(\lambda_{\beta})(|\xi_{\alpha_{\Lambda_k}}(A_{k,0}+A_{k,1})|+|\xi_{\alpha_{\Lambda_k}}(A_{k,0}-A_{k,1})|)
\Biggr\}^{\frac{1}{h}} \nonumber \\
=&\prod_{k\in\{r_1,r_2,...,r_{h}\}}
\Biggr(  2\sum_{\lambda_{\beta}\in\Lambda_k}p(\lambda_\beta)
{\rm
max}\{\xi_{\alpha_{\Lambda_k}}(A_{k,0}
),\xi_{\alpha_{\Lambda_k}}(A_{k,1})\}\Biggr)^{\frac{1}{h}} \nonumber \\
\leq&2,\end{align}
where ${k'}$ denotes the scenario where the party sharing sources with the independent party ${\bf A}_k$ has ${\bf A}_{k'}$, and the measurements performed by 
${\bf A}_{k'}$ on this subsystem with inputs 0 and 1 are denoted as 
$A_{k',0}$ and 
$A_{k',1}$, respectively,
$X_i$ 
  denotes all measurements with input $i$ performed on the remaining non-independent parties in the network, excluding the subsystems connected to the independent parties,  the second inequality holds because of the H\"older inequality, and the last obeys the conditions of order-preserving of state $\tau$ and the spectra of $A_{i,0}$ and $A_{i,1}$ are $-1$ and $1$ for any $k\in\{r_1,r_2,...,r_{h_1}\}$. 
 \end{proof}

Apparently, according to Theorems \ref{the3.2} and \ref{the3.3}, it can be seen that the violation of the inequality depends not only on the algebraic structure but also on the choice of observables and states. Next, we define a quantity 
\begin{align} \label{g319}
\mathcal S(\tau,\mathcal M_{A_1}, \mathcal M_{A_2},\cdots \mathcal M_{A_m})
=\sup_{\{A_{i,0}, A_{i,1}\}_{i\in\{1,2,...,m\}}}(|I_\tau|^{\frac{1}{h}}+|J_\tau|^{\frac{1}{h}}).
\end{align}

From Theorems \ref{the3.1} and \ref{the3.2}, we naturally obtain the following corollary.

\begin{corollary} In the mutually commuting von Neuamnn algebra models of the quantum network $\Xi(n,m)$ with the maximal independent number $h_{\rm max}$, 

{\rm (1)} $2\leq \mathcal S(\tau,\mathcal M_{A_1}, \mathcal M_{A_2},\cdots,\mathcal M_{A_m})\leq2\sqrt{2}$

{\rm (2)} If $\mathcal M_{A_{r_1}}, \mathcal  M_{A_{r_2}},...,\mathcal M_{A_{r_h}}$ are Abelian, then 
$$\mathcal S(\tau,\mathcal M_{A_1}, \mathcal M_{A_2},\cdots \mathcal M_{A_m})=2,$$
 where $2\leq h\leq h_{\rm max}$ and $\{r_1,r_2,...,r_h\}\in\Xi(n,m,h^{D_h})$.

{\rm (3)} For any states $\phi,\psi\in[(\mathcal M_{A_1} \vee \mathcal M_{A_2}\vee \cdots \vee\mathcal M_{A_m})^{\prime\prime}]^*$, the following inequality holds:
$|\mathcal S(\phi, \mathcal M_{A_1}, \mathcal M_{A_2},\cdots,\mathcal M_{A_m})-\mathcal S(\psi, \mathcal M_{A_1}, \mathcal M_{A_2},\cdots,\mathcal M_{A_m})|\le k\|\phi-\psi\|^{\frac{1}{h}}$, where $k$ is a positive constant. 
So, the functional $\phi\mapsto\mathcal S(\phi, \mathcal M_{A_1}, \mathcal M_{A_2},\cdots,\mathcal M_{A_m})=\sup_{\{A_{i,0}, A_{i,1}\}_{i\in\{1,2,\cdots,m\}}}(|I_\tau|^{\frac{1}{h}}+|J_\tau|^{\frac{1}{h}})$ is norm continuous.

\end{corollary}

\begin{proof}
To prove  (1). Since $|I_\tau|^{\frac{1}{h}}+|J_\tau|^{\frac{1}{h}}=2$ when $A_{i,0}=A_{i,1}=A_{j,0}=I$ for any $i\in\{r_1,r_2,\dots,r_h\}$ and $j\in\{1,2,\dots,m\}\setminus\{r_1,r_2,\dots,r_h\}$. 
Combining the Theorem \ref{the3.1}, we get (1).

(2) holds by Theorem \ref{the3.2} and the fact that $|I_\tau|^{\frac{1}{h}}+|J_\tau|^{\frac{1}{h}}=2$ when $A_{i,0}=A_{i,1}=A_{j,0}=I$ for any $i\in\{r_1,r_2,\dots,r_h\}$ and $j\in
\{1,2,\dots,m\}\setminus\{r_1,r_2,\dots,r_h\}$.

To prove (3). By combining $|\sup x - \sup y| \le \sup |x-y|$ and the triangle inequality, one can get
\begin{align}
&\Biggr| \sup_{\{A_{i,0}, A_{i,1}\}_{i\in\{1,2,...,m\}}}(|I_\phi|^{\frac{1}{h}}+|J_\phi|^{\frac{1}{h}})-\sup_{\{A_{i,0}, A_{i,1}\}_{i\in\{1,2,...,m\}}}(|I_\psi|^{\frac{1}{h}}+|J_\psi|^{\frac{1}{h}})\Biggr|\nonumber \\
=&\Biggr|\sup_{\{A_{i,0}, A_{i,1}\}_{i\in\{1,2,...,m\}}}\Biggr\{
\Biggl|\phi\Biggl(\prod_{i\in\{r_1,r_2,\dots,r_h\}}(A_{i,0}+A_{i,1})\prod_{j\in\{1,2,\dots,m\}\setminus\{r_1,r_2,\dots,r_h\}}A_{j,0}\Biggr) \Biggr|^{\frac{1}{h}} \nonumber \\
&+\Biggl|\phi\Biggl(\prod\limits_{i\in\{r_1,r_2,\dots,r_h\}}(A_{i,0}-A_{i,1})\prod_{j\in\{1,2,\dots,m\}\setminus\{r_1,r_2,\dots,r_h\}}A_{j,1}\Biggr)\Biggr|^{\frac{1}{h}} \Biggl\}\nonumber\\
&-\sup_{\{A_{i,0}, A_{i,1}\}_{i\in\{1,2,...,m\}}}\Biggr\{
\Biggl|\psi\Biggl(\prod_{i\in\{r_1,r_2,\dots,r_h\}}(A_{i,0}+A_{i,1})\prod_{j\in\{1,2,\dots,m\}\setminus\{r_1,r_2,\dots,r_h\}}A_{j,0}\Biggr) \Biggr|^{\frac{1}{h}} \nonumber \\
&+\Biggl|\psi\Biggl(\prod\limits_{i\in\{r_1,r_2,\dots,r_h\}}(A_{i,0}-A_{i,1})\prod_{j\in\{1,2,\dots,m\}\setminus\{r_1,r_2,\dots,r_h\}}A_{j,1}\Biggr)\Biggr|^{\frac{1}{h}} \Biggl\} \Biggl|\nonumber\\
\leq&\sup_{\{A_{i,0}, A_{i,1}\}_{i\in\{1,2,...,m\}}}\Biggr| 
\Biggl|\phi\Biggl(\prod_{i\in\{r_1,r_2,\dots,r_h\}}(A_{i,0}+A_{i,1})\prod_{j\in\{1,2,\dots,m\}\setminus\{r_1,r_2,\dots,r_h\}}A_{j,0}\Biggr) \Biggr|^{\frac{1}{h}} \nonumber \\
&+\Biggl|\phi\Biggl(\prod\limits_{i\in\{r_1,r_2,\dots,r_h\}}(A_{i,0}-A_{i,1})\prod_{j\in\{1,2,\dots,m\}\setminus\{r_1,r_2,\dots,r_h\}}A_{j,1}\Biggr)\Biggr|^{\frac{1}{h}}
\nonumber \\
&-\Biggl|\psi\Biggl(\prod_{i\in\{r_1,r_2,\dots,r_h\}}(A_{i,0}+A_{i,1})\prod_{j\in\{1,2,\dots,m\}\setminus\{r_1,r_2,\dots,r_h\}}A_{j,0}\Biggr) \Biggr|^{\frac{1}{h}} \nonumber \\
&-\Biggl|\psi\Biggl(\prod\limits_{i\in\{r_1,r_2,\dots,r_h\}}(A_{i,0}-A_{i,1})\prod_{j\in\{1,2,\dots,m\}\setminus\{r_1,r_2,\dots,r_h\}}A_{j,1}\Biggr)\Biggr|^{\frac{1}{h}}
\Biggl| 
\nonumber \\
\leq&\sup_{\{A_{i,0}, A_{i,1}\}_{i\in\{1,2,...,m\}}}
\Biggl\{ \Biggl|\Biggl|\phi\Biggl(\prod_{i\in\{r_1,r_2,\dots,r_h\}}(A_{i,0}+A_{i,1})\prod_{j\in\{1,2,\dots,m\}\setminus\{r_1,r_2,\dots,r_h\}}A_{j,0}\Biggr) \Biggr|^{\frac{1}{h}}
\nonumber \\
&-\Biggl|\psi\Biggl(\prod_{i\in\{r_1,r_2,\dots,r_h\}}(A_{i,0}+A_{i,1})\prod_{j\in\{1,2,\dots,m\}\setminus\{r_1,r_2,\dots,r_h\}}A_{j,0}\Biggr) \Biggr|^{\frac{1}{h}} \Biggl|\nonumber \\
&+\Biggl|\Biggl|\phi\Biggl(\prod\limits_{i\in\{r_1,r_2,\dots,r_h\}}(A_{i,0}-A_{i,1})\prod_{j\in\{1,2,\dots,m\}\setminus\{r_1,r_2,\dots,r_h\}}A_{j,1}\Biggr)\Biggr|^{\frac{1}{h}}
\nonumber \\
&-\Biggl|\psi\Biggl(\prod\limits_{i\in\{r_1,r_2,\dots,r_h\}}(A_{i,0}-A_{i,1})\prod_{j\in\{1,2,\dots,m\}\setminus\{r_1,r_2,\dots,r_h\}}A_{j,1}\Biggr)\Biggr|^{\frac{1}{h}}\Biggl|
\Biggr\} \nonumber \\
\leq&\sup_{\{A_{i,0}, A_{i,1}\}_{i\in\{1,2,...,m\}}}
\Biggl\{  
\Biggl|(\phi-\psi)\Biggl(\prod_{i\in\{r_1,r_2,\dots,r_h\}}(A_{i,0}+A_{i,1})\prod_{j\in\{1,2,\dots,m\}\setminus\{r_1,r_2,\dots,r_h\}}A_{j,0}\Biggr) \Biggr|^{\frac{1}{h}}
\nonumber \\
&+\Biggl|(\phi-\psi)\Biggl(\prod_{i\in\{r_1,r_2,\dots,r_h\}}(A_{i,0}-A_{i,1})\prod_{j\in\{1,2,\dots,m\}\setminus\{r_1,r_2,\dots,r_h\}}A_{j,1}\Biggr) \Biggr|^{\frac{1}{h}}
\Biggr\} \nonumber \\
\leq&\sup_{\{A_{i,0}, A_{i,1}\}_{i\in\{1,2,...,m\}}}
\Biggl\{  \Biggl(
||\phi-\psi|| \Biggl| \Biggl|\prod_{i\in\{r_1,r_2,\dots,r_h\}}(A_{i,0}+A_{i,1})\prod_{j\in\{1,2,\dots,m\}\setminus\{r_1,r_2,\dots,r_h\}}A_{j,0}\Biggr|\Biggl|\Biggl)^{\frac{1}{h}}
\nonumber \\
&+\Biggr( ||\phi-\psi||\Biggl|\Biggl|\prod_{i\in\{r_1,r_2,\dots,r_h\}}(A_{i,0}-A_{i,1})\prod_{j\in\{1,2,\dots,m\}\setminus\{r_1,r_2,\dots,r_h\}}A_{j,1}\Biggr) \Biggr|\Biggl|\Biggr) ^{\frac{1}{h}}
\Biggr\} \nonumber \\
\leq&k||\phi-\psi||^{\frac{1}{h}}.
\end{align}
The reason why the third inequality holds here is that the function $f(x)=x^{\frac{1}{h}}$ is concave and monotonically increasing, and satisfies the H\"older continuity condition, so 
$|a^{\frac{1}{h}}-b^{\frac{1}{h}}|\leq|a-b|^{\frac{1}{h}}$ holds for any $a,b\geq0$ and 
$h\geq1$. Thus, the functional $\phi\rightarrow\mathcal S(\phi, \mathcal M_{A_1}, \mathcal M_{A_2},\cdots,\mathcal M_{A_m})=\sup_{\{A_{i,0}, A_{i,1}\}_{i\in\{1,2,\cdots,m\}}}(|I_\tau|^{\frac{1}{h}}+|J_\tau|^{\frac{1}{h}})$ is norm continuous.
\end{proof}

\section{Maximal violation of Bell-type inequalities and algebraic structures}

In the quantum network $\Xi(n,m)$, this section examines how the violation of Bell-type inequalities, particularly their maximal violation, can reflect structural properties of the underlying algebra. A violation of a Bell-type inequality means $\mathcal S_{\tau} > 2$, and the maximum quantum violation is $2\sqrt{2}$. In the following theorem, we analyze the conditions under which maximal violation occurs.

\begin{theorem}\label{thm4.1}For the mutually commuting von Neuamnn algebra models of the quantum network $\Xi(n,m)$ with the maximal independent number $h_{\rm max}$, if $\tau\in \mathcal M_{A_1A_2\cdots A_m}^*$ is a faithful network state and $\{r_1,r_2,...,r_h\}\in\Gamma(n,m,h^{D_h})$, then  
\begin{align}
|I_\tau|^{\frac{1}{h}}+|J_\tau|^{\frac{1}{h}}=2\sqrt{2}
\end{align}
if and only if $\tau((A_{i,0})^2A_i)=\tau(A_i),\tau((A_{i,1})^2A_i)=\tau(A_i)$ and $\tau[(A_{i,0}A_{i,1}+A_{i,1}A_{i,0})A_i]=0$  for any $A_i\in\mathcal{M}_{A_i}$,  $i\in\{r_1,r_2,...,r_h\}$, and $\tau(\prod_{j\in\{1,2,\dots,m\}\setminus\{r_1,r_2,\dots,r_h\}}(A_{j,0})^2A)=\tau(\prod_{j\in\{1,2,\dots,m\}\setminus\{r_1,r_2,\dots,r_h\}}(A_{j,1})^2A)=\tau(A)$ for any $A$ belonging to the von Neumann algebra generated by the non-independent parties.
\end{theorem}
\begin{proof}
If $|I_\tau|^{\frac{1}{h}}+|J_\tau|^{\frac{1}{h}}=2\sqrt{2}$, according to the proof process of Theorem \ref{the3.1}, the conditions that the operators corresponding to each party should satisfy are as follows:

(1) To ensure the equality in Ineq. (\ref{E1}) holds, condition 
\begin{align}
\pi_{\tau}\Biggl(\prod_{j\in\{1,2,\dots,m\}\setminus\{r_1,r_2,\dots,r_h\}}A_{j,0}\Biggr)\Omega=k_0\pi_{\tau}\Biggl(\prod_{i\in\{r_1,r_2,\dots,r_h\}}(A_{i,0}+A_{i,1})\Biggr)\Omega,\label{eq4.2} \\
\pi_{\tau}\Biggl(\prod_{j\in\{1,2,\dots,m\}\setminus\{r_1,r_2,\dots,r_h\}}A_{j,1}\Biggr)\Omega=k_1\pi_{\tau}\Biggl(\prod_{i\in\{r_1,r_2,\dots,r_h\}}(A_{i,0}-A_{i,1})\Biggr)\Omega \label{eq4.3}
\end{align}
must be satisfied.

(2) For the equality in the Ineq. (\ref{E2}) to hold,
condition 
\begin{align}
&\Biggl|\Biggl|\pi_{\tau}\Biggl(\prod_{j\in\{1,2,\dots,m\}\setminus\{r_1,r_2,\dots,r_h\}}A_{j,0}\Biggr)\Omega\Biggr|\Biggr|=1, \label{eq4.4}\\
&\Biggl|\Biggl| \pi_{\tau}\Biggl(\prod_{j\in\{1,2,\dots,m\}\setminus\{r_1,r_2,\dots,r_h\}}A_{j,1}\Biggr)\Omega\Biggr|\Biggr|=1\label{eq4.5}
\end{align}
must be satisfied.

(3) 
To ensure the equality in Ineqs. (\ref{g35}) and (\ref{g36}) holds, condition 
\begin{align} \label{E5}
(A_{i,0})^2+(A_{i,1})^2=2I 
\end{align}
must obviously hold for any $i\in\{r_1,r_2,\dots,r_h\}$.

(4) For the Ineq. (\ref{g38}), when condition $\theta_{r_1}=\theta_{r_2}=...=\theta_{r_h}$ holds, the equality holds.  Let $\theta_{r_1}=\theta_{r_2}=...=\theta_{r_h}=\theta$. This means that $\tau(A_{i,0}A_{i,1}+A_{i,1}A_{i,0})=4{\rm sin}^2\theta-2=\tau(X)$ for any $i\in\{r_1,r_2,\dots,r_h\}$.

(5) 
To ensure the equality in Ineq. (\ref{E3}) holds, condition 
\begin{align}\label{E47}
\tau\Biggl[\prod_{i\in\{r_1,r_2,\dots,r_h\}}(A_{i,0}+A_{i,1})^2\Biggr]=\tau\Biggl[\prod_{i\in\{r_1,r_2,\dots,r_h\}}(A_{i,0}-A_{i,1})^2\Biggr]
\end{align}
must hold. It follows from Eqs. (\ref{E5}) and  (\ref{E47})
that  we get 
\begin{align}\label{E8}
\tau(A_{i,0}A_{i,1}+A_{i,1}A_{i,0})=0
\end{align}
for any $i\in\{r_1,r_2,\dots,r_h\}$.

(6)  It follows from H\"older inequality that the equality in Ineq. (\ref{E3}) holds when
\begin{align}
\tau[(A_{r_1,0}+A_{r_1,1})^2]=\lambda_2(\tau[(A_{r_2,0}+A_{r_2,1})^2])=\cdots=\lambda_{h}(\tau[(A_{r_h,0}+A_{r_h,1})^2]), \label{eq4.6}\\
\tau[(A_{r_1,0}-A_{r_1,1})^2]=\lambda_2(\tau[(A_{r_2,0}-A_{r_2,1})^2])=\cdots=\lambda_{h}(\tau[(A_{r_h,0}-A_{r_h,1})^2]). \label{eq4.7}
\end{align}
It follows from Eqs. (\ref{E5}) and (\ref{E8}) that we can obtain 
\begin{align}
&\lambda_2=
\lambda_3=\cdots=\lambda_h=1, \label{eq4.9} \\
&\tau(A_{u,0}A_{u,1}+A_{u,0}A_{u,1})=\tau(A_{v,0}A_{v,1}+A_{v,0}A_{v,1})=\tau(X)\label{eq4.10}
\end{align}
for any $u,v\in\{i_1,i_2,...,i_h\}$.

(7) Equality in the Ineq. (\ref{E7}) implies that one can get $(A_{i,0})^2=(A_{i,1})^2=I$ because $-I\le A_{i,0},A_{i,1}\le I$ for any $i\in\{r_1,r_2,...,r_h\}$, therefore $\tau((A_{i,0})^2A_i)=\tau(A_i)$ and  $\tau((A_{i,1})^2A_i)=\tau(A_i)$ for any $A_i\in\mathcal M_{A_i}$.

Combining the Eqs. (\ref{eq4.2})-(\ref{eq4.5}), one can get
\begin{align}
k_0^2\tau\Biggl[\prod_{i\in\{r_1,r_2,\dots,r_h\}}(A_{i,0}+A_{i,1})^2\Biggr]=1, \\
k_1^2\tau\Biggl[\prod_{i\in\{r_1,r_2,\dots,r_h\}}(A_{i,0}-A_{i,1})^2\Biggr]=1.
\end{align}
Using the fact of Eq. (\ref{eq4.10}), we have
\begin{align}
|k_0|=\frac{1}{[2+\tau(X)]^{\frac{h}{2}}}, \label{eq4.13}\\
|k_1|=\frac{1}{[2-\tau(X)]^{\frac{h}{2}}}. \label{eq4.14}
\end{align}
It follows from Eqs. (\ref{eq4.2})-(\ref{eq4.5}), (\ref{eq4.13}) and (\ref{eq4.14}) that we get 
\begin{align}
 \frac{1}{(2+\tau(X))^{\frac{h}{2}}}(2+\tau(X))^{h}=
  \frac{1}{(2-\tau(X))^{\frac{h}{2}}}(2-\tau(X))^{h}.
\end{align}
This means that  $\tau(X)=0$ and 
$k_0^2=k_1^2=\frac{1}{2^h}$.   Obviously, regarding result $\tau({X})$, it is consistent with Eq. (\ref{E8}).
 
Since 
\begin{align}
&\Biggl|\Biggl\langle\Omega,\pi_\tau\Biggl(\prod_{j\in\{1,2,\dots,m\}\setminus\{r_1,r_2,\dots,r_h\}}(A_{j,0})^2\Biggr)\Omega\Biggr\rangle \Biggr| \nonumber \\
=&\Biggl|\frac{1}{2^h}\Biggl\langle\Omega,
\pi_\tau\Biggl(
\prod_{i\in\{r_1,r_2,\dots,r_h\}}(A_{i,0}+A_{i,1})^2
\Biggr)\Omega\Biggr\rangle\Biggr| \nonumber \\
\le&\frac{1}{2^h}\sqrt{\Biggl|\Biggl| \pi_\tau\Biggl(
\prod_{i\in\{r_1,r_2,\dots,r_h\}}(A_{i,0}+A_{i,1})
\Biggr)\Omega\Biggr|\Biggr|}\sqrt{\Biggl|\Biggl| \pi_\tau\Biggl(
\prod_{i\in\{r_1,r_2,\dots,r_h\}}(A_{i,0}+A_{i,1})
\Biggr)\Omega\Biggr|\Biggr|}\nonumber \\
=&\frac{1}{2^h}\prod_{i\in\{r_1,r_2,\dots,r_h\}}\tau[(A_{i,0}+A_{i,1})^2]\nonumber \\
=&1,
\end{align}
and
\begin{align}
\tau\Biggl(\prod_{j\in\{1,2,\dots,m\}\setminus\{r_1,r_2,\dots,r_h\}}(A_{j,0})^2\Biggr)=
\Biggl|\Biggl\langle\Omega,\pi_\tau\Biggl(\prod_{j\in\{1,2,\dots,m\}\setminus\{r_1,r_2,\dots,r_h\}}(A_{j,0})^2\Biggr)\Omega\Biggr\rangle \Biggr|=1.
\end{align}
It means that $\pi_\tau(\prod_{j\in\{1,2,\dots,m\}\setminus\{r_1,r_2,\dots,r_h\}}(A_{j,0})^2)\Omega=\Omega$. Then, we get 
\begin{align}
&\tau(\prod_{j\in\{1,2,\dots,m\}\setminus\{r_1,r_2,\dots,r_h\}}(A_{j,0})^2A)=\langle\Omega, \pi_{\tau}(\prod_{j\in\{1,2,\dots,m\}\setminus\{r_1,r_2,\dots,r_h\}}(A_{j,0})^2)\pi_{\tau}(A)\Omega\rangle
\nonumber \\
=&
\langle \pi_{\tau}(\prod_{j\in\{1,2,\dots,m\}\setminus\{r_1,r_2,\dots,r_h\}}(A_{j,0})^2)\Omega, \pi_{\tau}(A)\Omega\rangle \nonumber \\
=&\langle \Omega, \pi_{\tau}(A)\Omega\rangle\nonumber \\
=&\tau(A).
\end{align}
Similarly, we can also obtain
$\tau(\prod_{j\in\{1,2,\dots,m\}\setminus\{r_1,r_2,\dots,r_h\}}(A_{j,1})^2A)=\tau(A)$ for any $A$ belonging to the von Neumann algebra generated by the non-independent parties. Furthermore, according to $\pi_\tau(\prod_{j\in\{1,2,\dots,m\}\setminus\{r_1,r_2,\dots,r_h\}}(A_{j,0})^2)\Omega=\Omega$ and Eq. (\ref{eq4.2}), we can obtain that 
\begin{align}
    \pi_{\tau}(f(X_{r_1},X_{r_2},...,X_{r_h}))\Omega=0
\end{align}
holds. Here $f(X_{r_1},X_{r_2},...,X_{r_h})=\prod_{i=1}^h(2I+X_{r_i})-2^hI$ with $X_{r_i}=A_{r_i,0}A_{r_i,1}+A_{r_i,1}A_{r_i,0}$.
Thus, $\tau[(f(X_{r_1},X_{r_2},...,X_{r_h}))^2]=0$. We can infer that $\tau(X_{r_i}^2)=0$ for any $i\in\{1,2...,h\}$.

Combining the faithfulness, non-negativity of the state $\tau$ and the self-adjointness of $X_{r_i}$ for any $i\in\{1,2...,h\}$. Then $X_{r_i}=0$, i.e., $$A_{r_i,0}A_{r_i,1}+A_{r_i,1}A_{r_i,0}=0,$$ and implies that $\tau[(A_{r_i,0}A_{r_i,1}+A_{r_i,1}A_{r_i,0})A_{r_i}]=0$.

To prove sufficiency, we can calculate the left-hand side of Ineq. (\ref{eq3.2}) when $\tau((A_{i,0})^2A_i)=\tau(A_i),\tau((A_{i,1})^2A_i)=\tau(A_i)$ and $\tau[(A_{i,0}A_{i,1}+A_{i,1}A_{i,0})A_i]=0$  for any $A_i\in\mathcal{M}_{A_i}$,  $i\in\{r_1,r_2,...,r_h\}$, and $\tau(\prod_{j\in\{1,2,\dots,m\}\setminus\{r_1,r_2,\dots,r_h\}}(A_{j,0})^2A)=\tau(\prod_{j\in\{1,2,\dots,m\}\setminus\{r_1,r_2,\dots,r_h\}}(A_{j,1})^2A)=\tau(A)$ for any $A$ belongs to the von Neumann algebra generated by the non-independent parties.  Obviously, $|I_\tau|^{\frac{1}{h}}+|J_\tau|^{\frac{1}{h}}=2\sqrt{2}$ can be obtained.

\end{proof}

Based on the conclusion of Theorem \ref{thm4.1},  we want to obtain the algebraic structure of the system when the Bell-type inequality achieves its maximal violation. In the Corollary \ref{co4.2}, we can more intuitively illustrate the requirements imposed on the algebraic structure when the Bell-type inequality reaches its maximal quantum violation.

\begin{corollary}\label{co4.2}
For the mutually commuting von Neuamnn algebra models of the quantum network $\Xi(n,m)$ with the maximal independent number $h_{\rm max}$, the Bell-type inequality can be maximally violated if and only if $\mathcal M_{A_{r_1}}, \mathcal M_{A_{r_2}},...,\mathcal M_{A_{r_h}}$ contain subalgebras isomorphic to \( M_2(\mathbb C) \) and there exists a faithful state
$\tau\in \mathcal M_{A_1A_2\cdots A_m}^*$
that satisfies
\begin{align}\label{con4}
\tau(A_{r_1}A_{r_2}\cdots A_{r_h})=\tau(A_{r_1})\tau(A_{r_2})\cdots \tau(A_{r_h}),
\end{align}
for all $(A_{r_1}, A_{r_2}, \cdots A_{r_h})\in \mathcal{M}_{A_{r_1}}\times \mathcal{M}_{A_{r_2}} \cdots \times\mathcal{M}_{A_{r_h}}$. Here, $2\leq h \leq h_{\rm max}$  and $\{r_1,r_2,...,r_h\}
\in \Gamma(n,m,h^{D_h})$.
\end{corollary}
\begin{proof}
$(\Leftarrow)$  For any $i\in\{r_1,r_2,...,r_h\}$, suppose $\mathcal{M}_{A_{i}}$  contains a subalgebra $\mathcal{M}_{A_{i}}^{sub}\simeq M_2(\mathbb{C})$  and $\tau\in \mathcal M_{A_1A_2\cdots A_m}^*$ satisfies Eq. (\ref{con4}). 
Then there exist operators $A_{i,0},A_{i,1},$ and $A_{i,2}:=-\frac{{\rm{i}}}{2}[A_{i,0},A_{i,1}]$ in $\mathcal{M}_{A_i}^{sub}$ such that they anticommute and $(A_{i,j})^2=I$ for $j\in\{0,1,2\}$, where ${\rm{i}}^2=-1$. Consequently, we obtain $$\tau((A_{i,0})^2A_i)=\tau(A_i), \,\,\tau((A_{i,1})^2A_i)=\tau(A_i),\ \tau[(A_{i,0}A_{i,1}+A_{i,1}A_{i,0})A_i]=0$$ for any $A_i\in\mathcal{M}_{A_i}$. 

Setting 
$(\prod_{j\in\{1,2,\dots,m\}\setminus\{r_1,r_2,\dots,r_h\}}A_{j,0})^2=(\prod_{j\in\{1,2,\dots,m\}\setminus\{r_1,r_2,\dots,r_h\}}A_{j,1})^2=I$ gives $$\tau(\prod_{j\in\{1,2,\dots,m\}\setminus\{r_1,r_2,\dots,r_h\}}(A_{j,0})^2A)=\tau(\prod_{j\in\{1,2,\dots,m\}\setminus\{r_1,r_2,\dots,r_h\}}(A_{j,1})^2A)=\tau(A)$$  for any $A$ belonging to the von Neumann algebra generated by the non-independent parties.

 By Theorem \ref{thm4.1}, these operators yield the maximal violation $2\sqrt{2}$ of the quantity $\mathcal{S}_\tau$ in Eq.\,(\ref{biin}). 

$(\Rightarrow)$ Conversely, assume $\mathcal{S}_\tau$ in Eq.\,(\ref{biin}) attains the maximal violation $2\sqrt{2}$. It follows from Theorem \ref{thm4.1} that we have
$\tau((A_{i,0})^2A_i)=\tau(A_i),\tau((A_{i,1})^2A_i)=\tau(A_i)$ and $\tau[(A_{i,0}A_{i,1}+A_{i,1}A_{i,0})A_i]=0$ for any $A_i\in\mathcal{M}_{A_i}$ and $i\in\{r_1,r_2,...,r_h\}$.  Furthermore, $\tau(\prod_{j\in\{1,2,\dots,m\}\setminus\{r_1,r_2,\dots,r_h\}}(A_{j,0})^2A)=\tau(\prod_{j\in\{1,2,\dots,m\}\setminus\{r_1,r_2,\dots,r_h\}}(A_{j,1})^2A)=\tau(A)$ for any $A$ belongs to the von Neumann algebra generated by the non-independent parties.

Taking $A_i=A_{i,0}A_{i,1}+A_{i,1}A_{i,0}$, it follows from the faithfulness of state $\tau$ that  $\tau[(A_{i,0}A_{i,1}+A_{i,1}A_{i,0})^2]=0$, one gets $$A_{i,0}A_{i,1}+A_{i,1}A_{i,0}=0,$$ i.e., $A_{i,0}A_{i,1}=-A_{i,1}A_{i,0}$. The algebra generated by $A_{i,0}, A_{i,1}$ is $$\mathcal{U}(A_{i,0}, A_{i,1}):=\{\sum_k \alpha_k A_{i,0}^sA_{i,1}^t|\alpha_k\in\mathbb C,s,t\in\mathbb N\}.$$ 
From $\tau(A_{i,0}^2A_i)=\tau(A_{i,1}^2A_i)=\tau(A_i)$ and $-I\le A_{i,0}, A_{i,1}\le I$, setting $A_i=I$ gives $\tau(I-(A_{i,0})^2)=\tau(I-(A_{i,1})^2)=0$. Faithfulness of state $\tau$ then implies $(A_{i,0})^2=(A_{i,1})^2=I$.  
Because $(A_{i,0})^2=(A_{i,1})^2=I$, one gets $\sum_k \alpha_k A_{i,0}^sA_{i,1}^t=\alpha_0 I+\alpha_1 A_{i,0}+\alpha_2A_{i,1}+\alpha_3A_{i,0}A_{i,1}$.
Hence,
$$\mathcal{U}(A_{i,0},A_{i,1})={\rm span}\{I, A_{i,0},A_{i,0},-\frac{{\rm i}}{2}[A_{i,0},A_{i,1}]\}\simeq M_2(\mathbb C).$$ 
So this proof is completed.
\end{proof}

Obviously, Corollary \ref{co4.2} further clarifies the von Neumann algebra structure corresponding to the independent parties in the MCvNA model of the quantum network when the maximal violation of a Bell-type inequality is achieved, and that the state can be faithful at this point.  Meanwhile, it can also be seen that the necessary and sufficient conditions for achieving the maximal quantum violation at this point impose no restrictions on the algebraic structures corresponding to the non-independent parties. It is worth noting that, since the choice of independent parties is not unique, one can select as many sets of independent parties as possible and determine the structure of certain von Neumann algebras by achieving the maximal violation of  Ineq. (\ref{biin}). The relationship can be more clear in the tensor product algebra model. 

\begin{corollary}\label{co4.3}
For the tensor product algebra models of the quantum network $\Xi(n,m)$ with the maximal independent number $h_{\rm max}$,  assume that each $\mathcal M_{A_{r_1}}$ is factor. Then the Bell-type inequality cannot be maximally violated if and only if at least one of $\mathcal M_{A_{r_1}}, \mathcal M_{A_{r_2}},...,\mathcal M_{A_{r_h}}$ is type I$_{2k+1}$ for some  finite positive integer $k$.
\end{corollary}

Although in non-relativistic quantum information theory, when seeking maximal violations of Bell inequalities in networks, the main focus is always on finding numerically optimal measurements-a task that can be accomplished using optimization algorithms such as SDP- Corollary \ref{co4.3} indicates under what conditions such optimal measurements can exist.

\section{Conclusions and discussions}

In this paper, our main contribution is to establish a more abstract von Neumann algebra framework for quantum networks with arbitrary structure. This provides a new perspective for studying quantum information theory in more general contexts such as infinite-dimensional systems and quantum field theory, and adds a new research direction to quantum field theory and operator algebra theory. We discuss the violation of Bell-type inequalities in quantum networks under this model, obtaining conclusions that are completely different from those in non-relativistic scenarios. In summary, the violation of Bell-type inequalities here is determined by the structural properties of the algebra. 

However, these efforts are still just a small beginning. In the Bell nonlocality theory on models of mutually commuting von Neumann algebras, many related issues remain to be further explored. Furthermore, many concepts of quantum networks have the potential to be studied and developed within this von Neumann algebra framework. This will be promising.

\section*{Acknowledgements}
 This
work is supported by the National Natural Science
Foundation of China (Grant No. 12271394).


\end{document}